\documentclass{svproc}
\usepackage{url}

\usepackage{subfigure}
\usepackage{amssymb}
\usepackage{amsmath}
\usepackage{graphicx}
\usepackage{colortbl,dcolumn}
\usepackage{booktabs}
\usepackage{xcolor}
\usepackage{threeparttable}

\usepackage{algorithm}
\usepackage{algorithmic}

\usepackage{amsthm}
\newtheorem{assumption}{Assumption}

\newcommand{\aspref}[1]{Assumption~\ref{#1}}
\newcommand{\defref}[1]{Definition~\ref{#1}}
\newcommand{\lemref}[1]{Lemma~\ref{#1}}

\newcommand{\propref}[1]{Proposition~\ref{#1}}
\newcommand{\corref}[1]{Corollary~\ref{#1}}
\newcommand{\figref}[1]{Fig.~\ref{#1}}

\newcommand{\secref}[1]{Section~\ref{#1}}

\newcommand{\algref}[1]{Algorithm~\ref{#1}}

\usepackage{mathtools} % \DeclareDelimiter etc

\usepackage{enumerate}
\graphicspath{{pdfs/},{images/},{authorphotos/}}

\renewcommand{\tilde}{\widetilde}
\renewcommand{\hat}{\widehat}
\renewcommand{\bar}[1]{\overline{#1}}
\renewcommand{\[}{[\![ }
\renewcommand{\]}{]\!] }

\begin{document}
\mainmatter              % start of a contribution
\title{Cooperative Filtering with Range Measurements: A Distributed Constrained Zonotopic Method}
\titlerunning{Cooperative Filtering}  % abbreviated title (for running head)
\author{Yu Ding\inst{1} \and Yirui Cong\inst{1}
Xiangke Wang\inst{1} \and Long Cheng\inst{2,3} }
\authorrunning{Yu Ding et al.} % abbreviated author list (for running head)
\tocauthor{Yirui Cong, Xiangke Wang,  and Long Cheng}
\institute{College of Intelligence Science and Technology, National University of Defense Technology, Changsha 410073, China,\\
\email{xkwang@nudt.edu.cn},
%\\ WWW home page:
%\texttt{http://users/\homedir iekeland/web/welcome.html}
\and
School of Artificial Intelligence, University of Chinese Academy of Sciences, Beijing 100049, China.
\and
State Key Laboratory for Management and Control of Complex Systems, Institute of Automation, Chinese Academy of Sciences, Beijing 100190, China.}

\maketitle              % typeset the title of the contribution

\begin{abstract}
This article studies the distributed estimation problem of a multi-agent system with bounded absolute and relative range measurements.
Parts of the agents are with high-accuracy absolute measurements, which are considered as anchors;
the other agents utilize low-accuracy absolute and relative range measurements, each derives an uncertain range that contains its true state in a distributed manner.
Different from previous studies, we design a distributed algorithm to handle the range measurements based on extended constrained zonotopes, which has low computational complexity and high precision.
With our proposed algorithm, agents can derive their uncertain range sequentially along the chain topology, such that agents
with low-accuracy sensors can benefit from the high-accuracy absolute measurements of anchors and improve the estimation performance.
Simulation results corroborate the effectiveness of our proposed algorithm and verify our method can significantly improve the estimation accuracy.
\end{abstract}
\section{Introduction}\label{sec:Introduction}
In this work, we are interested in problems where parts of agents are equipped with high-accuracy absolute measurements, which are considered as anchors, while the other agents need to utilize low-accuracy absolute measurements and relative range measurements to derive an uncertain range that contains the true state of each agent in a distributed manner.

Recently, cooperative localization problems with absolute and relative measurements for multi-agent systems has received considerable research attention~\cite{2019Cooperative,2017Distributed}.
\begin{itemize}
\item
For Luenberger observer methods, in \cite{SelfLocalization2015}, a distributed self-localization algorithm was developed for multi-agent systems by using relative position measurements.
The results were extended to study the problem of distributed source localization by using bearing angle measurements in~\cite{2016Distributed}.
In \cite{2018Integrated}, distributed consensus filters were developed for a set of targets to utilize velocity and distance relative measurements to estimate the relative positions.
\item
For stochastic methods, as in \cite{Viegas2018}, distributed Kalman filters have been designed for a multi-agent system with linear  relative measurements based on the trace minimization of the covariance matrix. A similar problem has been investigated in \cite{2019distributed}, with different methods to derive suboptimal filter gains.
In~\cite{2018Cooperative}, nonlinear relative measurements have been considered in the cooperative filtering problem for multi-agent systems, where a recursive quadratic filter has been proposed to compute the exact value of the variance explicitly.
\end{itemize}

It is worth mentioning that the above-mentioned methods have limitations when facing unknown but bounded noises, which motivates us to focus on the guaranteed method, i.e., the set-membership filter (SMF), to solve the cooperative estimation problem of absolute and relative measurements.
To the best of our knowledge, few articles have investigated this problem.
In \cite{Silvestre2022}, the author proposed the constrained convex generator (CCG) as a generalization of the constrained zonotope, and applied  CCG to deal with the relative range measurement, which provides us much inspiration in this article.
However, the method proposed in \cite{Silvestre2022} only works with a known beacon, which implies its limitation facing the scene that anchors belong to an uncertain range.
In \cite{2023DSMF}, the first step of distributed SMF of taking general relative measurements into account has been taken, the authors provided a framework to derive uncertain range in a distributed manner, which is an important base of our work.
It should be pointed out that \cite{2023DSMF} mainly focused on the general framework analysis, but neglected the design of algorithms.
In \cite{2023SMFCooperative}, a distributed constrained zonotopic algorithm to  was developed based on the distributed framework proposed in \cite{2023DSMF}, but only linear relative measurements were considered.

Thus in this article, we study the set-membership estimation problem of multi-agent systems with absolute and relative range measurements, such that  agents with low absolute measurement accuracy sensors can benefit from relative range measurements with the agents with high accuracy absolute measurements.

The main contributions of this paper are summarized as follows:
\begin{itemize}
\item
We propose an extended constrained zonotopic algorithm to deal with the relative range measurement.
The proposed algorithm only utilizes local and neighborhood information, which can be  realized in a distributed manner.
\item
Based on our algorithm, agents with low-accuracy sensors can benefit from the high-accuracy absolute measurements of anchors along the chain topology, such that the estimation precisions can be efficiently improved.
\end{itemize}

\textit{Notations:}
For an uncertain variable $\mathbf{x} $, a realization  is defined as $\mathbf{x}(\omega)=x $.
The range of an uncertain variable is described by  its range:
$\[\mathbf{x} \]=\{ x(\omega):\omega \in \Omega \}$,
the conditional range of $\mathbf{x}$ given $\mathbf{y}=y$ is $\[ \mathbf{x}|y   \] $~\cite{Girish2013}.
A directed graph $\mathcal{G}=(\mathcal{V},\mathcal{E},\mathcal{A})$ is used to represent system topology.
 $\mathbb{N}_0 $ is the set of  natural number.
 $\mathbb{R}^n$ denotes the $n$-dimensional Euclidean space.
Given two sets $S_1$ and $S_2$ in an Euclidean space, the operation $\oplus  $ stands for the Minkowski sum of $S_1$ and $S_2$,
i.e., $S_1 \oplus S_2=\{s_1+s_2:s_1 \in S_1,~s_2 \in S_2 \}$.
The diameter of a set $S$ is defined by $d(S)=\underset{s_1,s_2 \in S}{\mathrm{sup}}  \| s_1-s_2  \|$.

%
%\subsection{Paper Organization}
The rest paper is organized as follows:
\secref{sec:System_Model} gives the system model and the problem description.
Main results are presented in \secref{sec:Main_Results}, with the distributed SMFing framework analyzed in \secref{sec:SMFing Framework Analysis} and constrained zonotopic algorithm designed in \secref{sec:State Estimation With Range/Bearing Data Based on Constrained Zonotope}.
Simulation examples  are provided in \secref{sec:Simulation} to corroborate our theoretical results.
Finally, \secref{sec:Conclusion} gives the concluding remarks of the paper.

\section{System Model and Problem Description}\label{sec:System_Model}
Consider a  multi-agent system, where each agent is identified by a  positive integer $i \in \{1,2,\ldots,N \}$ and described by the following dynamics equation:
%
%The dynamic of agent $i$  is represented by a discrete-time Linear Time Invariant (LTI) model
\begin{align}\label{eqn:agent_dynamic}
x_{i,k+1}=A_i x_{i,k}+B_i u_{i,k}+ w_{i,k},
\end{align}
where $k \in \mathbb{N}_0 $.
As a realization of $\mathbf{x}_{i,k}$, $x_{i,k} \in \[ \mathbf{x}_{i,k} \] \subseteq \mathbb{R}^{n_i}$ denotes the state at time instant $k$ for agent $i$;
$u_{i,k}$ denotes the input of agent $i$ at step $k$, which is assumed as a known signal in this article.
$w_{i,k} \in \[ \mathbf{w}_{i,k} \] \subseteq \mathbb{R}^{n_i}$ is the process noise.
%
%

%$f_i:D \subseteq \mathbb{R}^{n_i} \to \mathbb{R}^{n_i}  $ stands for the state transition function.
%, which is considered to be locally lipschitz.

In this work, we consider two types of measurements, i.e., absolute measurements and relative measurements, which are widely considered in the literature \cite{2019distributed,2018Cooperative}.

\textbf{Absolute measurement:}
An absolute measurement refers to agent $i$ takes an observation on its state directly, with an measurement equation as the following form:
\begin{align}\label{eqn:absolute_measurement}
y_{i,k}=C_i x_{i,k} +v_{i,k},
\end{align}
where $y_{i,k}\in \mathbb{R}^{m_i}$ is the measurement at time instant $k$ for agent $i$;
$C_i \in \mathbb{R}^{n_i \times m_i}  $ is the measurement matrix, and the pair $(C_i,A_i)$ is observable;
 $v_{i,k} \in \[ \mathbf{v}_{i,k}  \] $ is the measurement noise,  with $d( \[\mathbf{v}_{i,k}\] ) \leq \bar{v}_i$, and $d(\cdot)$ returns the diameter of a set.

Note that the absolute measurement is accessible to all the agents, but with different precisions, i.e., parts of agents are equipped with  high accuracy absolute sensors.
These agents is considered as  `\textbf{anchors}'(or  `\textbf{beacons}'), which is common in cooperative localization studies~\cite{2009Distributed,Tian2007A}.
The other agents, in this article, are called \textbf{`non-anchors'} or \textbf{`original agents'}, which only possess low precision absolute measurements, and need to utilize relative measurements with their neighborhoods to improve estimation accuracy.

%bound $\| v\| \leq  V_i , \forall v \in \[ \mathbf{v}_{i,k}  \]$.
\textbf{Relative measurements:}
Relative  measurements are related to agent $i$ and its neighbors,
\begin{align}\label{eqn:relative_measurement}
y_{i,j,k}^r= \| x_{i,k}-x_{j,k}\|_2+r_{i,j,k},~j\in \mathcal{N}_i,
\end{align}
where $y_{i,j,k}^r\in \mathbb{R}^{p_i}$ is the measurement at time instant $k$ between agent $i$ and its in-neighborhood agent $j$;
$g_{i,j}: \mathbb{R}^{n_{i}} \times \mathbb{R}^{n_{j}} \to \mathbb{R}^{p_i} $ is the  measurement function, in this article two types of relative measurements are considered, i.e., relative range measurement and relative bearing measurement.
where $\|\cdot \|_2 $ is the 2-norm of a vector, which represents the range-based measurements (e.g., UWB measurements).
$r_{i,j,k} \in \[  \mathbf{r}_{i,j,k}  \]$ is the range measurement noise, with $ \underline{r}_{i,j,k} \leq  r_{i,j,k} \leq \bar{r}_{i,j,k} $.
%$d(\[ \mathbf{r}_{i,j,k}\]) \leq \bar{r}_{i,j}$.

In this article, the uncertain variables of  noise at each time instant and the initial state of each agent are assumed to be unrelated to guarantee the non-stochastic Hidden Markov Model \cite{2021Rethinking}, i.e.,
\begin{assumption}[Unrelated Noises and Initial State]
$\forall k \in \mathbb{N}_0$, $\mathbf{w}_{i,0:k}$,$\mathbf{v}_{i,0:k}$,$\mathbf{r}_{i,j,0:k}$,
$\mathbf{x}_{i,0}$,$~\forall i,j =1,\ldots,N$ are unrelated.
\end{assumption}

The relative range measurement is accessible to all the agents in the system, and the communication topology is consistent with the time-invariant measurement topology, i.e., $(v_j,v_i )\in \mathcal{E}$ implies agent $i$ takes a relative measurement from agent $j$ and receives information from agent $j$ simultaneously.

Consider a general topology of the multi-agent system as \figref{fig:general_topology}, where the red nodes represent agents with high precision absolute measurements.
\begin{figure}
  \centering
  \includegraphics[width=0.6\columnwidth]{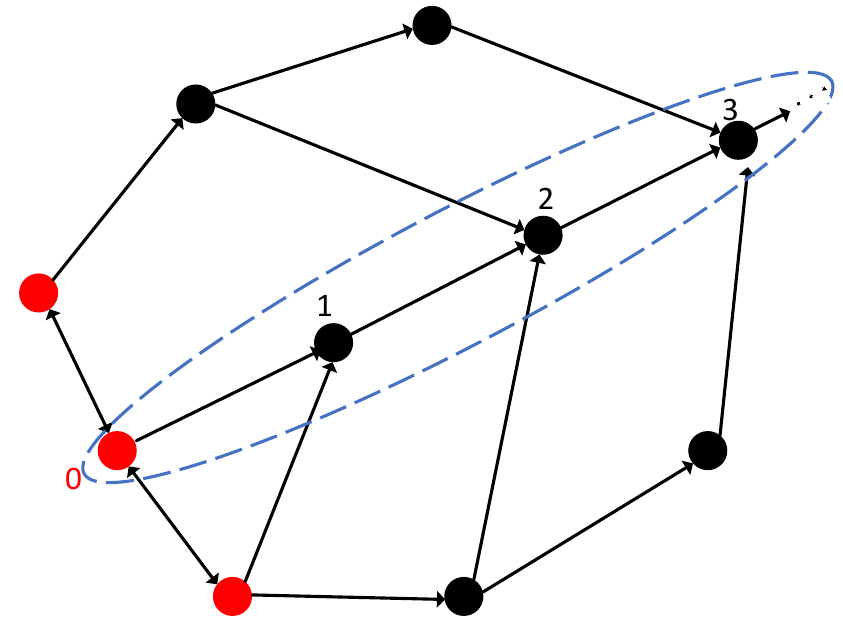}
  \caption{Diagram of chain topology. }\label{fig:general_topology}
\end{figure}

As in \cite{2013Understanding}, ordinary agents benefit from the absolute measurements of anchors through the chain, hence
without loss of generality, we consider  a subtopology of chain (circled by the blue dash line), and  other chains in the topology can be generalized in a similar way.
Without loss of generality, the anchor node is labeled as $0$, while other sensor nodes along the chain are labeled as $1,~2,~\cdots,l_{\mathrm{max}} $.

The aim of this paper is to derive an  uncertain range $S_{i,k}$  for each agent $i$ at each instant $k$ in a distributed manner, using all the measurements $y_{i,k}$ and $y_{i,i-1,k}^r$, $j \in \mathcal{N}_i$ up to $k$, such that
$ x_{i,k} \in S_{i,k}, \forall k \in \mathbb{N}_0$.
 %$
%

\section{Main Results}\label{sec:Main_Results}
In this section, we provide a distributed SMFing framework such that each agent  can derive its uncertain range sequentially along the chain topology (see \secref{sec:SMFing Framework Analysis}).
Then in \secref{sec:State Estimation With Range/Bearing Data Based on Constrained Zonotope}, a detailed algorithm based on the extend constrained zonotope is developed to realize the framework proposed in \secref{sec:SMFing Framework Analysis}.

\subsection{Sequentail SMFing Framework}\label{sec:SMFing Framework Analysis}

%Our analysis begins with a simple case, i.e., a topology of chain rooted by the anchor node (see \figref{fig:chain_topology}).
%

%
%

%For each agent $i$, the neighborhood agent in the chain topology is the agent $i-1$.
%
  According to the distributed SMFing framework proposed in \cite{2023DSMF}, the estimation process can be expressed as the following framework (see \algref{alg:Sequential_SMFing_Framework}).

\begin{algorithm}
\caption{Sequential  SMFing Framework}
\label{alg:Sequential_SMFing_Framework}
\begin{algorithmic}[1]
\STATE $\mathbf{Initialization}$. For each agent $i$, set the initial prior range $S_{i,0}$.
\LOOP
\STATE $\mathbf{Prediction}$. For $k \in \mathbb{N}_0$ ,given the posterior range $S_{i,k-1}$ in the previous step, the prior range $\bar{S}_{i,k} $ is predicted by
    \begin{align}
     \bar{S}_{i,k}=A_{i}(S_{i,k-1} ) \oplus B_i u_{i,k} \oplus  \[\mathbf{w}_{i,k-1} \].
     \end{align}
\STATE $\mathbf{Measurement}$. Take measurements $y_{0,k} $ and $y^r_{i,i-1,k}$.
\STATE $\mathbf{Update}$. The update is divided into $l_{\mathrm{max}}$ steps.
\STATE $i=0$, i.e., the update of anchor $0$
        \begin{align}\label{eqn:anchor_chain_update}
    S_{0,k}= \mathcal{X}_{0,k}(C_0, y_{0,k} , \[\mathbf{v}_{0,k}\])  \cap \bar{S}_{0,k},
    \end{align}
    where $\mathcal{X}_{0,k}(C_0,y_{0,k},\[ \mathbf{v}_{0,k} \])=\{ x_{0,k}: y_{0,k}=C_{0}x_{0,k}+v_{0,k},~v_{0,k} \in \[ \mathbf{v}_{0,k} \]  \}$.
\STATE  \textbf{While} $1 \leq i \leq l_{\mathrm{max}}$
      \begin{align}\label{eqn:uncertain_range_iteration}
       S_{i,k}=\bar{S}_{i,k} \cap g_{i,i-1}^{-1}(\{ y^r_{i,i-1,k} \}\oplus \[-\mathbf{r}_{i,i-1,k}  \] ,S_{i-1,k} ) 
              \cap \mathcal{X}_{i,k}(C_i,y_{i,k},\[ \mathbf{v}_{i,k} \]).
       \end{align}
\STATE   $i=i+1$;
\STATE End while.
\STATE  $k=k+1$;
\ENDLOOP
\end{algorithmic}
\end{algorithm}

It should be noticed that the posterior uncertain range of each agent $i$ is computed step by step from $0$ to $l_{\mathrm{max}}$, and hence the process shown in \algref{alg:Sequential_SMFing_Framework} is called ``sequential localization'' in cooperative localization problems.

It should be noticed that for relative range measurements  $g_{i,j}(x_{i,k},x_{j,k})=\| x_{i,k}-x_{j,k}   \|_2$, $g_{i,j}^{-1}$ has its specific geometrical significance.
\begin{proposition}
The measurement set derived by the relative range measurement satisfies:
\begin{align}
\{x_{i,k}: x_{i,k} \in g_{i,j}^{-1}(\{z_{i,j,k}\} \oplus \[ -\mathbf{r}_{i,j,k} \]  , S_{j,k} )\} =S_{j,k} \oplus Y^r_{i,j,k},
\end{align}
where $Y^r_{i,j,k}$ is a geometrical ring defined by $Y^r_{i,j,k}=\{s: \underline{R}_{i,j,k} \leq   \|  s\|_2 \leq \bar{R}_{i,j,k} \} $, $\underline{R}_{i,j,k}=y^r_{i,j,k}-\bar{r}_{i,j,k}$ and $ \underline{R}_{i,j,k}=y^r_{i,j,k}-\underline{r}_{i,j,k}$.
\end{proposition}

\begin{proof}
Let $s=x_{i,k}-x_{j,k} $, then we have $\| s \|_2+r_{i,j,k}=z_{i,j,k},~r_{i,j,k} \in \[ -\mathbf{r}_{i,j,k} \] $, which implies  $\| s \|_2 \in [\underline{R}_{i,j,k},~\bar{R}_{i,j,k} ]$ and $ Y^r_{i,j,k}=\bigcup \{ s \}.$
%\begin{align}
%& \| s \|_2 \in [\underline{R}_{i,j,k},~\bar{R}_{i,j,k} ], \\ \notag
%& Y^r_{i,j,k}=\bigcup \{ s \}.
%\end{align}

Moreover, there holds
\begin{align}
x_{i,k}=x_{j,k}+s,~x_{j,k} \in S_{j,k},~s\in Y^r_{i,j,k}.
\end{align}
From the definition of Minkowski sum, we have
\begin{align}
x_{i,k} \in S_{j,k} \oplus Y^r_{i,j,k},
\end{align}
which completes the proof.
\end{proof}

\subsection{Extended Constrained Zonotopic Algorithm Design}\label{sec:State Estimation With Range/Bearing Data Based on Constrained Zonotope}
In this section, the state estimation strategy is presented to over-approximate the measurement set.

\begin{definition}  (Extended Constrained Zonotope~\cite{2022YiruiCong})
\label{def:Extended_Constrained_Zonotope}
A set $\mathcal{Z} \subseteq \mathbb{R}^n$ is an extended constrained zonotope if there exists a quintuple $(G,c,A,b,h) \in \mathbb{R}^{ n \times n_g} \times \mathbb{R}^n \times \mathbb{R}^{ n_c \times n_g} \times \mathbb{R}^{ n_c} \times [0,\infty ]^{n_g} $ such that $\mathcal{Z}$ is expressed by
\begin{align}\label{eqn:Extended_Constrained_Zonotope}
\{ G \xi +c : A \xi =b, \xi \in \underset{j=1}{\overset{n_g}{\prod}} [-h^{(j)}, h^{(j)} ] \} =:Z(G,c,A,b,h),
\end{align}
where $ h^{(j)}$ is the $j$\textsuperscript{th} component of $h$.
\end{definition}
When $h^{(j)}=1,~j=1,\cdots,n_g$, \defref{def:Extended_Constrained_Zonotope} becomes the classical constrained zonotope in \cite{Scott2016}.

\begin{lemma}\label{lem:CZ_properties}
For every $Z=(G_z, c_z,A_z,b_z,h_z)\subset \mathbb{R}^n $, $W=(G_w,c_w,A_w,b_w,h_w) \subset \mathbb{R}^n$,
$Y=(G_y,c_y,A_y,b_y,h_y) \subset \mathbb{R}^k$ and $T \in  \mathbb{R}^{k \times n} $, the three identities hold:
\begin{align}\label{eqn:CZ_linear_transform}
TZ&=(TG_z,Tc_z,A_z,b_z,h_z ), \\
Z \oplus W&=([G_z~G_w],c_z+c_w,\begin{bmatrix}A_z &0 \\ 0 & A_w \end{bmatrix}, \begin{bmatrix}b_z \\ b_w \end{bmatrix} ,\begin{bmatrix}h_z \\ h_w \end{bmatrix}),\\
Z \cap_{\mathbf{T}}Y&=([G_z~0],c_z,\begin{bmatrix}A_z & 0 \\ 0 & A_y \\ TG_z & -G_y \end{bmatrix}, \begin{bmatrix}b_z \\ b_y \\ c_y-Tc_z \end{bmatrix} ,\begin{bmatrix}h_z \\ h_y \end{bmatrix} ).
\end{align}
\end{lemma}

The prior and  posterior  uncertain range of agent $i$ is denoted by
$ \bar{\mathcal{Z}}_{i,k}=(\bar{G}_{i,k},\bar{c}_{i,k},\bar{A}_{i,k},\bar{b}_{i,k},\bar{h}_{i,k}) $
and  $\mathcal{Z}_{i,k}= (\hat{G}_{i,k},\hat{c}_{i,k},\hat{A}_{i,k},\hat{b}_{i,k},\hat{h}_{i,k}) $, respectively.

\begin{proposition}
Let $R=\underset{q \in \mathcal{I}}{\bigcup}R^q $, the following equation holds
\begin{align}
(S_1 \oplus R) \cap S_2 = \underset{s\in S_1}{\bigcup}[\underset{q \in \mathcal{I}}{\bigcup}(\{s \} \oplus R^q  )\cap S_2  ].
\end{align}
\end{proposition}

\begin{proof}
This conclusion comes directly from the commutativity of set union and Minkowski sum, i.e.,
\begin{align}
(S_1 \oplus \underset{q \in \mathcal{I}}{\bigcup}R^q)= \underset{q \in \mathcal{I}}{\bigcup}(S_1 \oplus R^q).
\end{align}
\end{proof}

Hence to calculate $\hat{\mathcal{Z}}_{i-1,k} \oplus Y^r_{i,i-1,k}$, we first consider a fixed point $x_{i,i-1,k}^0 \in \hat{\mathcal{Z}}_{i-1,k}$.
Assume there is a relative bearing  measurement between $ x_{i-1,k}^0$ and agent $i$
\begin{align}
y^{b}_{i,i-1,k}=\mathrm{ang}(x_{i,k}-x_{i-1,k}^0)+r^b_{i,i-1,k},
\end{align}
where $r^b_{i,i-1,k} \in [\underline{r}^b_{i,i-1,k},\bar{r}^b_{i,i-1,k} ]$ is the relative bearing measurement noise.

\begin{lemma}\label{lem:relative_position_CZ}
The measurement set given $Y^{br}_{i,i-1,k} $ by the relative range measurement $y^r_{i,i-1,k}$ and relative bearing measurement $ y^{b}_{i,i-1,k}$ is outer bounded by
\begin{multline}\label{eqn:UWB_position_CZ}
Y_{i,i-1,k}^{br}\subseteq   \{z:
z=[G_{y^{br}}~0]\begin{bmatrix}\xi_1 \\ \xi_2 \end{bmatrix}+c_{y^{br}}, \\
 \begin{bmatrix} H^{br}G_{y^b} & (\sigma_1-\sigma_2) \end{bmatrix} \begin{bmatrix}\xi_1 \\ \xi_2 \end{bmatrix}
 =\sigma_1-Hc_{y^b},
 \begin{bmatrix}\xi_1 \\ \xi_2 \end{bmatrix}\in h_{y^b}\times [0,1]\},
\end{multline}
where
\begin{align}\label{eqn:segment_part_cZrepresentation}
&G_{y^{br}}= \begin{bmatrix}\mathrm{cos}(\bar{\theta}_{i,i-1,k}) &\mathrm{cos}(\underline{\theta}_{i,i-1,k} ) \\
%(z^b_{i,i-1,k}+\bar{r}^b_{i,i-1,k} )&\mathrm{cos}(z^b_{i,i-1,k}+\underline{r}^b_{i,i-1,k} ) \\
\mathrm{sin}(\bar{\theta}_{i,i-1,k} )& \mathrm{sin}(\underline{\theta}_{i,i-1,k} )\end{bmatrix}, 
%\mathrm{sin}(z^b_{i,i-1,k}+\bar{r}^b_{i,i-1,k} )& \mathrm{sin}(z^b_{i,i-1,k}+\underline{r}^b_{i,i-1,k} )\end{bmatrix}, \\
~c_{y^{br}}=x_{i-1,k}^0, \\ \notag
%[\breve{p}_{i-1,x,k}^{T}, \breve{p}_{i-1,y,k}^{T}]^T, \\ \notag
&H^{br}=\begin{bmatrix} \mathrm{cos}(\frac{\bar{\theta}_{i,i-1,k}+\underline{\theta}_{i,i-1,k}}{2} ) & \mathrm{sin}(\frac{\bar{\theta}_{i,i-1,k}+\underline{\theta}_{i,i-1,k}}{2}) \end{bmatrix}, \\ \notag
&\sigma_1^{br}=\underline{R}_{i,i-1,k} \mathrm{cos}(\frac{\bar{\theta}^b_{i,i-1,k}+\underline{\theta}^b_{i,i-1,k}}{2} )+H^{br}x_{i-1,k}^0, 
~\sigma_2^{br}=\bar{R}_{i,i-1,k}+H^{br}x_{i-1,k}^0,  \\ \notag
&h_{y^{br}}=[0, \frac{\bar{R}_{i,i-1,k}} {\mathrm{cos}(\frac{\bar{r}^b_{i,i-1,k}-\underline{r}^b_{i,i-1,k}}{2} )} ]\times [0, \frac{\bar{R}_{i,i-1,k}} {\mathrm{cos}(\frac{\bar{r}^b_{i,i-1,k}-\underline{r}^b_{i,i-1,k}}{2} )} ], \\ \notag
&\underline{\theta}_{i,i-1,k}=y_{i,i-1,k}^b-\bar{r}_{i,i-1,k}^b,~\bar{\theta}_{i,i-1,k}=y_{i,i-1,k}^b-\underline{r}_{i,i-1,k}^b.
\end{align}
\end{lemma}

\begin{proof}
According to the geometrical significance,  $Y_{i,i-1,k}^{br}=Y_{i,i-1,k}^{b} \cap Y^r_{i,i-1,k} $ is outer bounded by the trapezoid $\mathrm{ABCD}$, which is the intersection of triangle $\{z:z\in  \triangle \mathrm{COD} \}$ and  hyperplane
$\{ z: Hz=\begin{bmatrix}\mathrm{cos}(\bar{z}^b_{i,i-1,k}) & \mathrm{sin}(\bar{z}^b_{i,i-1,k}) \end{bmatrix}
%\begin{bmatrix}
z \in [\sigma_1,\sigma_2 ] \}$.

For the triangle $\triangle \mathrm{COD}$, each point can be covered by the linear combination of $\frac{\vec{\mathrm{OC}}}{\|  \mathrm{\vec{\mathrm{OD}}} \|}$ and $\frac{\vec{\mathrm{OD}}}{\|  \mathrm{OD} \|}$, hence $\triangle \mathrm{COD}$ can be expressed by the following zonotope:
\begin{multline}
\mathcal{Z}_{i,i-1,k}^b =\{ z= G_{y^{br}}\xi_1
% \begin{bmatrix} \zeta_1 \\ \zeta_2 \end{bmatrix} +
+x_{i-1,k}^0, \\
\xi_1 \in 
%\begin{bmatrix}  \breve{p}_{i-1,x,k} \\ \breve{p}_{i-1,y,k} \end{bmatrix},  
  [0, \frac{\bar{R}_{i,i-1,k}} {\mathrm{cos}(\frac{\bar{r}^b_{i,i-1,k}-\underline{r}^b_{i,i-1,k}}{2} )} ]   \times
[0, \frac{\bar{R}_{i,i-1,k}} {\mathrm{cos}(\frac{\bar{r}^b_{i,i-1,k}-\underline{r}^b_{i,i-1,k}}{2} )} ]
\}.
\end{multline}

Notice that point $A=x_{i-1,k}^0
%\begin{bmatrix}  \breve{p}_{i-1,x,k}+  \\ \breve{p}_{i-1,y,k} \end{bmatrix}
+ \underline{R}_{i,i-1,k} \begin{bmatrix} \mathrm{cos}(\bar{\theta}_{i,i-1,k} ) \\ \mathrm{sin}(\bar{\theta}_{i,i-1,k} )\end{bmatrix} $ and point
$C=x_{i-1,k}^0
%\begin{bmatrix}  \breve{p}_{i-1,x,k}+  \\ \breve{p}_{i-1,y,k} \end{bmatrix}
+ \bar{R}_{i,i-1,k} \begin{bmatrix} \mathrm{cos}(\underline{\theta}_{i,i-1,k} ) \\ \mathrm{sin}(\underline{\theta}_{i,i-1,k} )\end{bmatrix}$, $\sigma_1^{br}$ and $\sigma_2^{br}$ can be obtained by $A \in \{ z: Hz=\sigma_1^{br} \}$ and $C \in \{ z: Hz=\sigma_2^{br} \}$.  Write $[\sigma_1^{br}, \sigma_2^{br}] $ as $\mathcal{Z}_{\sigma}=\{\sigma:\sigma= (\sigma_2^{br}-\sigma_1^{br}) \xi_2 +\sigma_1^{br} ,~ \xi_2 \in [0,1]  \} $, \eqref{eqn:UWB_position_CZ} can be obtained by $Y_{i,i-1,k}^{br} \subseteq  \mathcal{Z}_{i,i-1,k}^b \cap_H^{br} \mathcal{Z}_{\sigma} $ with \lemref{lem:CZ_properties}.
\end{proof}

\begin{figure}[htbp]
  \centering
  \includegraphics[width=0.7\textwidth]{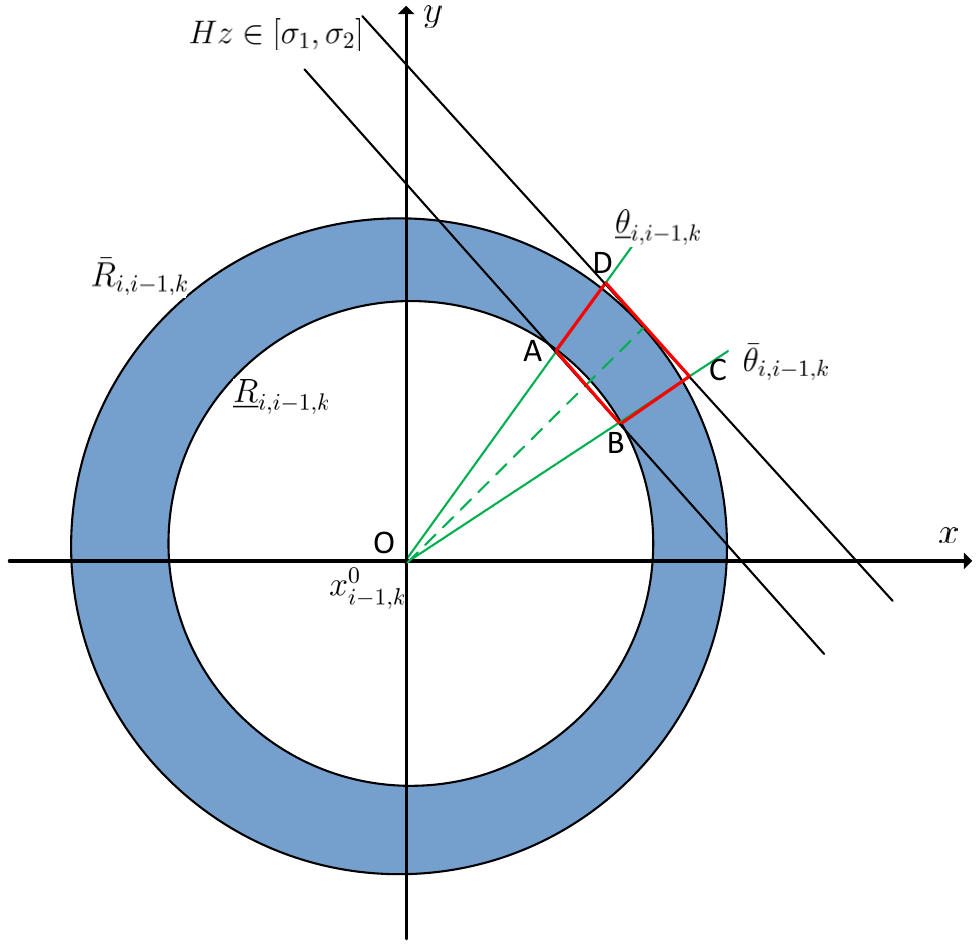}
  \caption{The diagram of proof of \lemref{lem:relative_position_CZ}.  }\label{fig:UWB_Position_Solution1}
\end{figure}

For only relative range measurement $y^r_{i,i-1,k}$, the set $Y^r_{i,i-1,k}$ can be approximated by segmenting  the full circle into $M$ parts, each part has the minimum angle value of $\frac{2 \pi q}{M}$ and maximum angle value $\frac{2 \pi (q+1)}{M},~q=0,1,\ldots,M-1$.
Repeat the computations above replacing the interval $[\underline{\theta}_{i,i-1,k}, \bar{\theta}_{i,i-1,k}]$
by $[\frac{2 \pi q}{M}, \frac{2 \pi (q+1)}{M}]$.
Each set $\mathcal{Z}_{i,i-1,k}^{0,r},\ldots,\mathcal{Z}_{i,i-1,k}^{M-1,r}$ needs to be intersected independently with the prediction prior uncertain range $\bar{S}_{i,k}$, and discarded if $\mathcal{Z}_{i,i-1,k}^{q,r} \cap \bar{S}_{i,k}=\emptyset$.
The range measurement set is a uncertain range correlated with the position of $x_{i-1,k}^0$, i.e.,
\begin{align}\label{eqn:range_measurement_center}
\mathcal{Z}^{r}_{i,i-1,k}(x_{i-1,k}^0)=\underset{\mathcal{Z}_{i,i-1,k}^{q,r} \cap \bar{S}_{i,k} \neq \emptyset}{\bigcup}\mathcal{Z}_{i,i-1,k}^{q,r}.
\end{align}

\begin{assumption}\label{asp:S_ik_relative_small}
The diameter of $\bar{\mathcal{Z}}_{i,k}$ satisfies
\begin{align}
d(\bar{\mathcal{Z}}_{i,k}) < \underline{R}_{i,i-1,k} < d(x_{i-1,k}^0,\bar{\mathcal{Z}}_{i,k} ).
\end{align}
\end{assumption}

\begin{proposition}\label{prop:leq_frac_pi3}
Under \aspref{asp:S_ik_relative_small}, if there exists $q_1, q_2 \in [0,M-1],~(q_1 < q_2) $ such that $ \mathcal{Z}_{i,i-1,k}^{q_1,r} \cap \bar{\mathcal{Z}}_{i,k} \neq \emptyset$ and $ \mathcal{Z}_{i,i-1,k}^{q_2,r} \cap \bar{\mathcal{Z}}_{i,k} \neq \emptyset$,
\begin{align}
 \mathcal{Z}_{i,i-1,k}^{q,r} \cap \bar{\mathcal{Z}}_{i,k} \neq \emptyset,~\forall   q \in [q_1,q_2].
\end{align}
Especially, let $q_1=\underset{q}{\mathrm{inf}} ~\mathcal{Z}_{i,i-1,k}^{q,r} \cap \bar{\mathcal{Z}}_{i,k} \neq \emptyset$ and
$q_2=\underset{q}{\mathrm{sup}} ~\mathcal{Z}_{i,i-1,k}^{q,r} \cap \bar{\mathcal{Z}}_{i,k} \neq \emptyset$, there holds
\begin{align}
\frac{2 \pi (q_2-q_1)}{M} < \frac{\pi}{3}.
\end{align}
\end{proposition}

\begin{proof}
Consider point $E \in \mathcal{Z}_{i,i-1,k}^{q_1,r} \cap \bar{\mathcal{Z}}_{i,k} $  and point $F \in \mathcal{Z}_{i,i-1,k}^{q_2,r} \cap \bar{\mathcal{Z}}_{i,k}$, and focus on the $\triangle \mathrm{OEF}$.
Notice that $\| EF\| \leq d(\bar{\mathcal{Z}}_{i,k})<\underline{R}_{i,i-1,k} $, according to the law of sines, $\angle \mathrm{EOF} < \frac{\pi}{3} $.

Moreover $\bar{\mathcal{Z}}_{i,k} $ is convex, hence   $\forall G \in \mathrm{EF} $, there holds $G \in \bar{\mathcal{Z}}_{i,k} $.
From $\underline{R}_{i,i-1,k} < d(x_{i-1,k}^0,\bar{\mathcal{Z} }_{i,k}) $  we have $\| OG \| \geq \underline{R}_{i,i-1,k} $.
Moreover, $\triangle \mathrm{EOF}$ is a oxygon, which implies
  $ \| OG \| < \| OF \| \leq \bar{R}_{i,i-1,k}$ and $G \in \mathcal{Z}_{i,i-1,k}^{q_2,r} \cap \bar{\mathcal{Z}}_{i,k},~\forall G \in \mathrm{EF}$.
   Hence
$\mathrm{EF} \subset \mathcal{Z}_{i,i-1,k}^{q,r} \cap \bar{\mathcal{Z}}_{i,k} $, \propref{prop:leq_frac_pi3} has been proved.
\end{proof}

\begin{figure}
  \centering
  \includegraphics[width=0.7\textwidth]{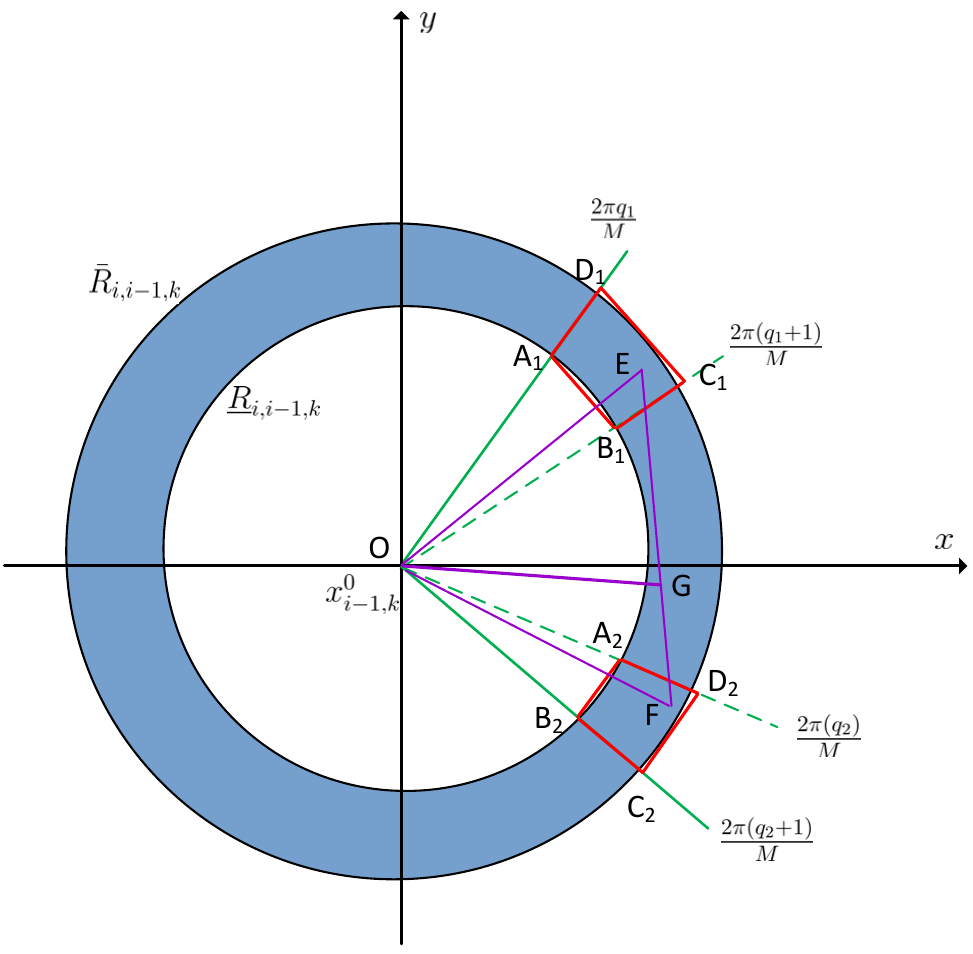}
  \caption{The diagram of proof of \propref{prop:leq_frac_pi3}.  }\label{fig:UWB_Position_Solution_sector_union}
\end{figure}

With \propref{prop:leq_frac_pi3}, the sectors can be united with the minimum phase angle $\frac{2 \pi q_1}{M} $ and maximum phase angle $\frac{2 \pi (q_2+1)}{M} $.

\begin{corollary}\label{cor:united_zonotope}
The  zonotope union
\begin{multline}
\underset{R_{i,i-1,k}^q(k) \cap \bar{S}_{i,k} \neq \emptyset}{\bigcup}R_{i,i-1,k}^q \subseteq  \mathcal{Z}^R_{i,i-1,k}(x_{i-1,k}^0)  
=(\hat{G}^R_{i,i-1,k}, \hat{c}^R_{i,i-1,k},\\
[H_q^R\hat{G}^R_{i,i-1,k} \quad (\sigma_2^R-\sigma_1^R)]
,\sigma_1^R-H_q^R x_{i-1,k}^0,~\hat{h}^R_{i,i-1,k} \times [0,1] ),
\end{multline}
where
\begin{align}
&\hat{G}^R_{i,i-1,k}= \begin{bmatrix}\mathrm{cos}(\frac{2 \pi (q_2+1)}{M} )&\mathrm{cos}(\frac{2 \pi q_1}{M} ) \\ \mathrm{sin}(\frac{2 \pi (q_2+1)}{M} )& \mathrm{sin}(\frac{2 \pi q_1}{M} )\end{bmatrix},
~\hat{c}^R_{i,i-1,k}=x_{i-1,k}^0, \\ \notag
&H_q^R=\begin{bmatrix} \mathrm{cos}(\frac{ \pi (q_1+q_2+1)}{M} ) & \mathrm{sin}(\frac{ \pi (q_1+q_2+1)}{M}) \end{bmatrix}, \\ \notag
&\sigma_1^R=\underline{R}_{i,i-1,k} \mathrm{cos}(\frac{(q_2-q_1+1) \pi }{M} )+H_q^R x_{i-1,k}^0,  
~\sigma_2^R=\bar{R}_{i,i-1,k}+H_q^R x_{i-1,k}^0,   \\ \notag
&h_{\tilde{R}^q}=[0, \frac{\bar{R}_{i,i-1,k}} {\mathrm{cos}(\frac{(q_2-q_1+1) \pi}{M} )} ]\times [0, \frac{\bar{R}_{i,i-1,k}} {\mathrm{cos}(\frac{(q_2-q_1+1) \pi}{M} )} ] .
\end{align}
\end{corollary}

It should be pointed out that $\mathcal{Z}^R_{i,i-1,k}(x_{i-1,k}^0)$ is correlated with the choice of $x_{i-1,k}^0 $, a different choice  $x_{i-1,k}^1 \neq x_{i-1,k}^0 $ would affect $q_1$ and $q_2$ (see \figref{fig:UWB_Position_Solution_center_correlation} ), hence calculating $\hat{\mathcal{Z}}_{i-1,k} \oplus Y^r_{i,i-1,k}$  directly is with huge computation burden.

%\begin{figure}[htbp]
%  \centering
%  \includegraphics[width=0.45\textwidth]{UWB_Position_Solution_center_correlation.pdf}
%  \caption{The calculation of $\mathcal{Z}_{i,i-1,k}^{q,r}$ relies on the center point choice of $x_{i-1,k} $.  }\label{fig:UWB_Position_Solution_center_correlation}
%\end{figure}
%\begin{figure}[htbp]
%  \centering
%  \includegraphics[width=0.4\textwidth]{UWB_Position_Solution_center_triangle.pdf}
%  \caption{The diagram of equivalence measurement noise.  }\label{fig:UWB_Position_Solution_center_triangle}
%\end{figure}

\begin{assumption}\label{asp:S_i_1k_relative_small}
The diameter $d(\hat{\mathcal{Z}}_{i-1,k} )$ satisfies
\begin{align}
d(\hat{\mathcal{Z}}_{i-1,k} ) \ll y_{i,i-1,k}.
\end{align}
\end{assumption}
\begin{remark}
\aspref{asp:S_ik_relative_small} and \aspref{asp:S_i_1k_relative_small}  imply $x_{i-1,k}$ is far from $x_{i,k} $,  this assumption is reasonable since the relative distance between agents is often much larger than the uncertain range caused by noise.
%when $d(x_{i-1,k}^0,\bar{\mathcal{Z}}_{i,k}) \gg d(\bar{\mathcal{Z}}_{i,k}) $,
\end{remark}

\begin{theorem}\label{thm:equivalent_measurement}
Under  \aspref{asp:S_ik_relative_small} and \aspref{asp:S_i_1k_relative_small}, there holds
\begin{align}
 (\hat{\mathcal{Z}}_{i-1,k} \oplus Y^r_{i,i-1,k} ) \cap \bar{\mathcal{Z}}_{i,k}
%\underset{x_{i-1,k}\in \hat{\mathcal{Z}}_{i-1,k} }{\bigcup} \mathcal{Z}_{i,i-1,k}^R(x_{i-1,k} ),
 \subseteq \tilde{\mathcal{Z}}_{i,i-1,k}^R(x_{i-1,k}^* )
\end{align}
where
\begin{multline}\label{eqn:final_extended_cZ}
\tilde{\mathcal{Z}}_{i,i-1,k}^R(x_{i-1,k}^* ) =\{
z=[\tilde{G}_{i,i-1,k}^R~0]\begin{bmatrix}\xi_1 \\ \xi_2 \end{bmatrix}+\tilde{c}_{i,i-1,k}^R, \\
 \begin{bmatrix} \tilde{H}\tilde{G}_{i,i-1,k}^R & (\tilde{\sigma}_1-\tilde{\sigma}_2) \end{bmatrix}
  \begin{bmatrix}\xi_1 \\ \xi_2 \end{bmatrix}
 =\tilde{\sigma}_1-\tilde{H} \tilde{c}_{i,i-1,k}^R,
 \begin{bmatrix}\xi_1 \\ \xi_2 \end{bmatrix}\in \tilde{h}_{i,i-1,k}^R\times [0,1]\},
\end{multline}
and  $\tilde{q}_1 $, $\tilde{q}_2$ are defined in \propref{prop:leq_frac_pi3}, replacing $\bar{R}_{i,i-1,k} $ and $\underline{R}_{i,i-1,k}$ with
\begin{align}
\tilde{\bar{R}}_{i,i-1,k}=d(\hat{\mathcal{Z}}_{i-1,k})+\bar{R}_{i,i-1,k} ,\quad
\tilde{\underline{R}}_{i,i-1,k}=\underline{R}_{i,i-1,k}-d(\hat{\mathcal{Z}}_{i-1,k}) 
\end{align}
and
\begin{align}\label{eqn:final_cZ_parameters}
&\tilde{G}_{i,i-1,k}^R= \begin{bmatrix}\mathrm{cos}(\frac{2 \pi (\tilde{q}_2+1)}{M} )&\mathrm{cos}(\frac{2 \pi \tilde{q}_1}{M} ) \\ \mathrm{sin}(\frac{2 \pi (\tilde{q}_2+1)}{M} )& \mathrm{sin}(\frac{2 \pi \tilde{q}_1}{M} )\end{bmatrix};
~\tilde{c}_{i,i-1,k}^R=\forall x_{i-1,k}^* \in \hat{\mathcal{Z}}_{i-1,k}, \\ \notag
&\tilde{\sigma}_1=(\tilde{\underline{R}}_{i,i-1,k}) \mathrm{cos}(\frac{(\tilde{q}_2-\tilde{q}_1+1) \pi }{M} ) +\tilde{H}x_{i-1,k}^*,~\tilde{\sigma}_2=\tilde{\bar{R}}_{i,i-1,k}+\tilde{H}x_{i-1,k}^*, \\ \notag
&\tilde{H}=\begin{bmatrix} \mathrm{cos}(\frac{ \pi (\tilde{q}_1+\tilde{q}_2+1)}{M} ) & \mathrm{sin}(\frac{ \pi (\tilde{q}_1+\tilde{q}_2+1)}{M}) \end{bmatrix}, \\ \notag
&\tilde{h}_{i,i-1,k}^R=[0, \frac{\tilde{\bar{R}}_{i,i-1,k}} {\mathrm{cos}(\frac{(\tilde{q}_2-\tilde{q}_1+1) \pi}{M} )} ]\times [0, \frac{\tilde{\bar{R}}_{i,i-1,k}} {\mathrm{cos}(\frac{(\tilde{q}_2-\tilde{q}_1+1) \pi}{M} )} ].
\end{align}
\end{theorem}

\begin{proof}
For a given point $x_{i,k} \in \bar{\mathcal{Z}}_{i,k}$, there holds
\begin{multline}\label{eqn:range_measurement_sup}
   \|x_{i-1,k}^*-x_{i,k}  \|_2
\leq \|x_{i-1,k}^*-x_{i-1,k}  \|_2 + \|x_{i-1,k}-x_{i,k}  \|_2  \\
 \leq d(\hat{\mathcal{Z}}_{i-1,k}) +\bar{R}_{i,i-1,k} =\tilde{\bar{R}}_{i,i-1,k}.
\end{multline}
Similarly, $\tilde{\underline{R}}_{i,i-1,k}=\underline{R}_{i,i-1,k}-d(\hat{\mathcal{Z}}_{i-1,k}) $ and
$\|x_{i-1,k}^*-x_{i,k}  \|_2 \geq \tilde{\underline{R}}_{i,i-1,k}$
%\begin{align}\label{eqn:range_measurement_inf}
%\end{align}
%
holds. Hence, $x_{i,k} \in \{s:\tilde{\underline{R}}_{i,i-1,k} \leq \| s\|_2 \leq  \tilde{\bar{R}}_{i,i-1,k}   \} $.
The rest proof is  same with \lemref{lem:relative_position_CZ} and \corref{cor:united_zonotope}.
\end{proof}

\begin{figure}[htbp]
\centering
\subfigure[]{\includegraphics [width=0.48\columnwidth, trim = 0 0 0 0]{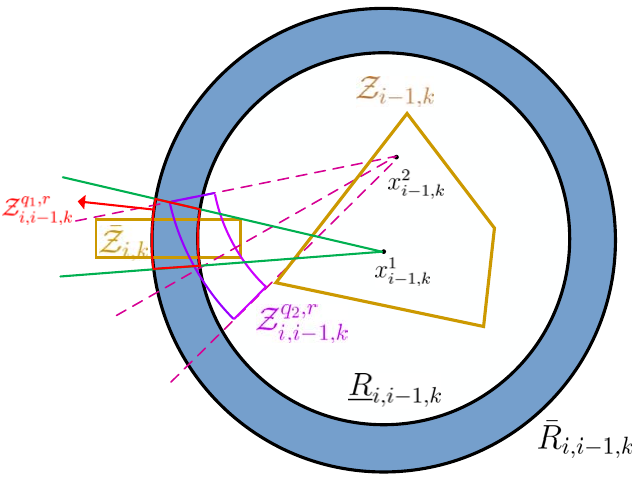}\label{fig:UWB_Position_Solution_center_correlation}}
\subfigure[]{\includegraphics [width=0.48\columnwidth, trim = 0 0 0 0]{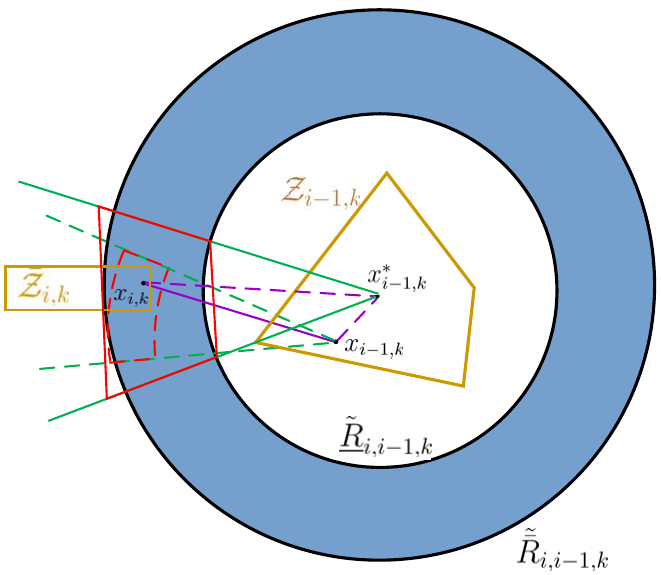}\label{fig:UWB_Position_Solution_center_triangle}}
\caption{
(a) The calculation of $\mathcal{Z}_{i,i-1,k}^{q,r}$ relies on the center point choice of $x_{i-1,k} $.
(b) The diagram of equivalence relative range measurement noise.
}\label{fig:UAV_Position_Estimation}
\end{figure}

To summarize, given the posterior uncertain range  $\hat{\mathcal{Z}}_{i-1,k} $ of agent $i-1$, prior uncertain range $\bar{\mathcal{Z}}_{i,k} $ and measurements $y_{i,k} $, $y_{i,i-1,k}^r$, the posterior extended constrained zonotope $\hat{\mathcal{Z}}_{i,k}$ (line 7 in \algref{alg:Sequential_SMFing_Framework}) can be calculated as follows:

\begin{algorithm}
\caption{Posterior Constrained Zonotope Calculation for Agent i}
\label{alg:Sequential_CZ_algorithm}
\begin{algorithmic}[1]
\STATE Let $x_{i-1,k}^*=\hat{c}_{i-1,k} $, then the diameter of $ \hat{ \mathcal{Z}}_{i-1,k}$ is given by $d(\hat{ \mathcal{Z}}_{i-1,k})=\underset{\|\xi\| \leq 1 }{\mathrm{sup}} \| \hat{G}_{i-1,k} \xi\|_{\infty} \leq \| \hat{G}_{i-1,k} \|_{\infty}$.
\STATE The equivalent relative range measurement set is given by the ring $\tilde{Y}^r_{i,i-1,k}=\{s: \tilde{\underline{R}}_{i,i-1,k} \leq  \| s \|_2   \leq \tilde{\bar{R}}_{i,i-1,k} \} $, where $\tilde{\underline{R}}_{i,i-1,k}=y^r_{i,i-1,k}-\bar{r}_{i,i-1,k}-d(\hat{ \mathcal{Z}}_{i-1,k}) $ and
$\tilde{\bar{R}}_{i,i-1,k}=y^r_{i,i-1,k}-\underline{r}_{i,i-1,k}+d(\hat{ \mathcal{Z}}_{i-1,k}) $.
\STATE Segment $\tilde{Y}^r_{i,i-1,k} $ into $M$ parts. Each part is outer bounded by the extended constrained zonotope $Z_{i,i-1,k}^{q,r}$ as in \eqref{eqn:segment_part_cZrepresentation}.
\LOOP
\STATE $Z_{i-i-1,k}^{q,r} \cap \bar{\mathcal{Z}}_{i,k} ,~q \in [0,M] $.
\STATE $q=q+1$.
\ENDLOOP
\STATE Let $q_1=\underset{q}{\mathrm{inf}} ~ Z_{i-i-1,k}^{q,r} \cap \bar{\mathcal{Z}}_{i,k} \neq \emptyset  $ and
$q_2=\underset{q} {\mathrm{sup} }~Z_{i-i-1,k}^{q,r} \cap \bar{\mathcal{Z}}_{i,k} \neq \emptyset  $, the extended constrained zonotope $\tilde{\mathcal{Z}}_{i,i-1,k}^R$ given by the relative range measurement $y_{i,i-1,k}^r$ can be represented as \eqref{eqn:final_extended_cZ} and \eqref{eqn:final_cZ_parameters}.
\STATE Calculate $\mathcal{Z}_{i,k}=\bar{\mathcal{Z}}_{i,k} \cap \tilde{\mathcal{Z}}_{i,i-1,k}^R \cap \mathcal{X}_{i,k} (C_i,y_{i,k},\[\mathbf{v}_{i,k} \] ) $ with \eqref{eqn:CZ_linear_transform}.
\STATE Calculate the interval hull of  $\mathcal{Z}_{i,k} $ as the posterior uncertain range $\hat{\mathcal{Z}}_{i,k}$ of agent $i$.
\end{algorithmic}
\end{algorithm}

\section{Simulation Results}\label{sec:Simulation}
In this section, we consider four UAVs in a 2-D plane with the  topology shown in \figref{fig:chain_topology}.
 UAV1 is the anchor node, while UAVs~2,~3, and 4 are equipped with low precision absolute sensors.
The state of each UAV is  $x_{i,k}=[p_{i,k,x}^T,p_{i,k,y}^T,\tilde{v}_{i,k,x}^T,\tilde{v}_{i,k,y}^T ]^T$.\footnote{To distinguish with the absolute measurement noise $v_{i,k}$, in this section we use   tilde $\tilde{v}_{i,k} $ to denote the velocity of UAV $i$ at time instant $k$.}
The aim of this simulation is to verify the uncertain range derived by our algorithm contains the real position of each UAV.
\begin{figure}
  \centering
  \includegraphics[width=0.6\columnwidth]{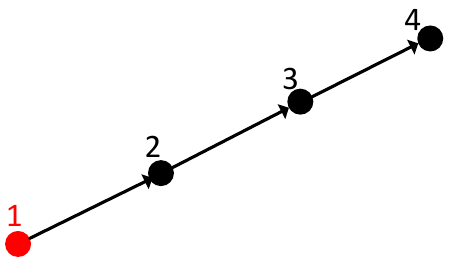}
  \caption{Simulation topology. }\label{fig:chain_topology}
\end{figure}
The parameters are settled as follows:
\begin{align}
A_i&=\begin{bmatrix}1 & T \\ 0 &1 \end{bmatrix} \otimes I_2 , B_i=\begin{bmatrix}\frac{T^2}{2} \\ T \end{bmatrix} \otimes \mathbf{1}_2,C_i=I_4\footnote{$\otimes $ denotes the kronecker product. }.
%A_i=\begin{bmatrix}1 & 0 & T & 0 \\ 0 &1 & 0 & T \\ 0 & 0 & 1 & 0 \\ 0 & 0& 0& 1 \end{bmatrix},
\end{align}
where $T=0.5s$.
The process noise for each UAV is
\begin{align}
\[\mathbf{w}_{i,k} \]&=\{w=0.1*I_4\xi,~ \|\xi\|_{\infty}\leq 1   \},
\end{align}
while the  measurement noise are settled as:
\begin{align}
&\[\mathbf{v}_{1,k} \]=\{\mathbf{v}_1=0.1 *I_4\xi, ~ \|\xi\|_{\infty}\leq 1   \}; \\ \notag
&\[\mathbf{v}_{2,k} \]=\[\mathbf{v}_{3,k} \]=\[\mathbf{v}_{4,k} \]=\{v=I_4\xi, ~ \|\xi\|_{\infty}\leq 1   \}; \\ \notag
&\[\mathbf{r}_{i,i-1,k} \]=[-0.1,0.1],~i=2,3,4.
\end{align}
The initial states and uncertain ranges of each UAV are randomly generated by Matlab.

Simulation results of position estimation of UAV4 are presented in \figref{fig:UAV4_Position_Estimation}, where the black line is the real trajectory, with the dots representing the real positions at different time steps.
The true state at each time step is contained by  the   uncertain range (rectangles)  derived by our proposed algorithm, which implies the effectiveness of our method in deriving uncertain ranges along the chain.

To demonstrate our method effectively improves estimation accuracy, we make a Monte Carlo experiment of 50 times, and compare our proposed methods with pure absolute measurement filtering method.
To measure the conservativeness directly, we consider using the diameter of sets, i.e.,
\begin{align}\label{eqn:diameter_of_set}
d(S)=\underset{s_1,s_2 \in S  }{\mathrm{sup}}\|s_1-s_2  \|_{\infty}.
\end{align}

For an interval hull in 2-D plane, $d(S)$ is the maximum edge length, which is the twice of the  infinite norm of  $\hat{G}$ matrix.
Thus we compare $\|G\|_{\infty}$ of the interval hull derived by our proposed method with pure absolute measurement filtering, as shown in Tabel~1.
%\tabref{tab:Conservativeness_comparison_with_pure_absolute_measurement_filtering}.
%
\begin{table}[!htbp]
\centering
\label{tab:Conservativeness_comparison_with_pure_absolute_measurement_filtering}
\caption{Conservativeness comparison with pure absolute measurement filtering}
\vspace{5pt}
\begin{tabular}{l|ccc}
\hline
 &UAV2 &UAV3 &UAV4  \\
\hline
$\| G_{pro} \|_{\infty}/ \| G_{abs} \|_{\infty} $ &0.4928 &0.5695 &0.7066 \\
\hline
\end{tabular}
\end{table}

From Tabel~1
%\tabref{tab:Conservativeness_comparison_with_pure_absolute_measurement_filtering} 
we can see that for UAV2, 3 and 4, estimation accuracy improved by our proposed method is up to $50.72\% $, $44.05\% $ and $29.44\% $, respectively.

\begin{figure}[htbp]
  \centering
  \includegraphics[width=0.75\textwidth]{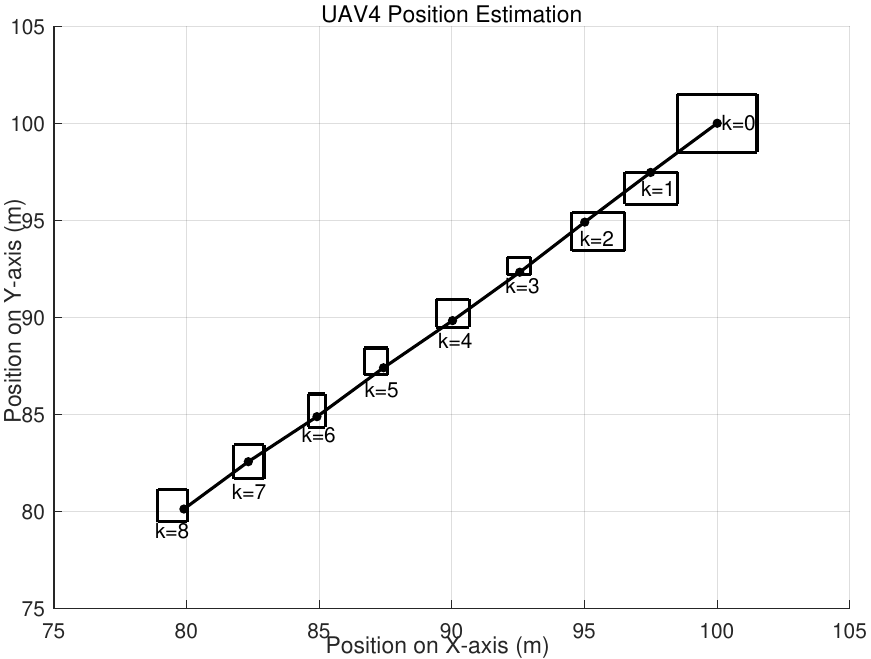}
  \caption{Position estimation simulation results with proposed algorithm.}\label{fig:UAV4_Position_Estimation}
\end{figure}

\section{Conclusion}\label{sec:Conclusion}

In this article, we have studied the distributed estimation problem of a multi-agent system with absolute and relative range measurements.
Parts of agents are equipped with high accuracy absolute measurements, while the other agents need to utilize low accuracy absolute measurements and relative range measurements to realize guaranteed estimation.
Firstly, we analyzed the set-membership filtering framework, where each agent in the chain topology can derive its uncertain range sequentially, such that agents with low accuracy can benefit from the high accuracy absolute measurements of anchors and improve estimation precisions.
Moreover, an algorithm based on the extended constrained zonotope has been developed to deal with the relative range measurements, such that each agent can derive its uncertain range in a distributed manner.
%we  proposed the distributed set-membership filtering framework for a system with chain topology.
%
%
%
Finally, several simulation results were presented, which  corroborate the effectiveness of our proposed algorithm  and verify our method achieves siginificant accuracy improvement in estimation.

For future work, we aim to focus on the estimation error boundedness analysis of our proposed algorithm.

%

%\paragraph{Notes and Comments.}
%The first results on subharmonics were
%obtained by Rabinowitz in \cite{fo:kes:nic:tue}, who showed the existence of
%infinitely many subharmonics both in the subquadratic and superquadratic
%case, with suitable growth conditions on $H'$. Again the duality
%approach enabled Clarke and Ekeland in \cite{may:ehr:stein} to treat the
%same problem in the convex-subquadratic case, with growth conditions on
%$H$ only.
%
%Recently, Michalek and Tarantello (see \cite{fost:kes} and \cite{czaj:fitz})
%have obtained lower bound on the number of subharmonics of period $kT$,
%based on symmetry considerations and on pinching estimates, as in
%Sect.~5.2 of this article.

%\bibliographystyle{IEEEtran}

%\bibliography{ArticleRegulation0222}

\end{document}